\documentclass{article}
\usepackage[keeplastbox]{flushend}  
\usepackage{spconf}
\usepackage{blindtext, graphicx}
\usepackage{bm}
\usepackage{booktabs}
\usepackage{subfig}
\usepackage[font=footnotesize]{caption}
\usepackage{hyperref}
\usepackage[warn]{textcomp}
\usepackage{siunitx}

\usepackage{algorithmic}
\usepackage[lined,ruled,noalgohanging,noend]{algorithm2e}

\usepackage{localmath}

\usepackage{nicefrac}
\newcommand{\half}{\nicefrac{1}{2}}

\usepackage{textcomp}

\makeatletter
\newcommand{\pushright}[1]{\ifmeasuring@#1\else\omit\hfill$\displaystyle#1$\fi\ignorespaces}
\newcommand{\pushleft}[1]{\ifmeasuring@#1\else\omit$\displaystyle#1$\hfill\fi\ignorespaces}
\makeatother

\usepackage[colorinlistoftodos,prependcaption,textsize=tiny]{todonotes}

\begin{document}
\ninept
\title{Multi-modal Blind Source Separation with Microphones and Blinkies}

\name{Robin Scheibler and Nobutaka Ono\thanks{This work was supported by a JSPS post-doctoral fellowship and grant-in-aid (\textnumero 17F17049), and the SECOM Science and Technology Foundation.}\thanks{The research presented in this paper is reproducible. Code and data are available at \protect\url{https://github.com/onolab-tmu/blinky-iva}.}%
\thanks{\textcopyright~2019 IEEE. Personal use of this material is permitted. Permission from IEEE must be obtained for all other uses, in any current or future media, including reprinting/republishing this material for advertising or promotional purposes,creating new collective works, for resale or redistribution to servers or lists, or reuse of any copyrighted component of this work in other works.}%
}
\address{Tokyo Metropolitan University, Tokyo, Japan}
%
%

\maketitle

\begin{abstract}
  We propose a blind source separation algorithm that jointly exploits measurements by a conventional microphone array and an ad hoc array of low-rate sound power sensors called \textit{blinkies}.
  While providing less information than microphones, blinkies circumvent some difficulties of microphone arrays in terms of manufacturing, synchronization, and deployment.
  The algorithm is derived from a joint probabilistic model of the microphone and sound power measurements.
  We assume the separated sources to follow a time-varying spherical Gaussian distribution, and the non-negative power measurement space-time matrix to have a low-rank structure.
  We show that alternating updates similar to those of independent vector analysis and Itakura-Saito non-negative matrix factorization decrease the negative log-likelihood of the joint distribution.
  The proposed algorithm is validated via numerical experiments.
  Its median separation performance is found to be up to 8 dB more than that of independent vector analysis, with significantly reduced variability.
\end{abstract}


\section{Introduction}

Blind source separation (BSS) conveniently allows to separate a mixture of sources without any prior knowledge about sources or microphones \cite{Comon:1512057}.
For example, independent component \cite{Comon:1994kr} and vector \cite{Kim:2006ex} analysis (ICA and IVA, respectively) reliably separate sources in the determined case, that is when there are as many microphones as sources.
The latter in particular cleverly avoids the frequency permutation ambiguity and can be solved with an efficient algorithm based on majorization-minimization~\cite{Ono:2011tn,Ono:2012wa}.
However, the recent drop in the cost of microphones and availability of plenty of processing power means that we are often in a situation where more microphones than sources are available. 
While more microphones should in principle lead to superior performance, algorithms designed for the determined case, such as IVA, may fail.
A typical problem is for a single source to have different frequency bands classified as different sources.

In this work, we explore the scenario where two modalities of sound, instantaneous pressure and short-time power, are collected with a compact microphone array and low-rate sound power sensors, respectively.
We assume that these sensors can be easily distributed in an ad hoc fashion in the area surrounding the target sound sources.
A practical example of such devices are \textit{blinkies}~\cite{scheibler2018apsipa_blinkies}.
These low-power battery operated sensors use a microphone to measure sound power which is used to modulate the brightness of an on-board light-emitting device (LED).
A conventional video camera is then used to synchronously harvest the measurements from  all blinkies.
This system is illustrated in \ffref{blinkies} along with an actual blinky device.
While the method presented hereafter is applicable to any device collecting sound power (e.g., distributed microphones, smartphones, etc), we will only refer to these sensors as blinkies for convenience in the rest of the paper.

Previous work has shown that a single blinky providing voice activity detection (VAD) of a single source can be used to create a powerful beamformer~\cite{scheibler2018apsipa_blinkies}.
This technique can leverage an arbitrary number of microphones, whose locations need not be known, and results in large improvements in source quality.
However, when several sound sources are present, only the power of their mixture can be measured, and the VAD becomes difficult to perform.
Moreover, errors in the VAD directly result in target source cancellation.
In this situation, non-negative matrix factorization (NMF) of the space-time sound power matrix has been proposed as a way of separating sources in the power domain \cite{horiike2018asj_fall}.
Such space-time NMF has also been suggested for noise suppression in asynchronous microphone arrays~\cite{Matsui:bg}.
Nevertheless, it remains an issue to find an appropriate threshold for the VAD following the NMF.

Instead of this two-step process, we propose to perform the source separation and the sound power NMF jointly.
Our approach builds on prior work showing that IVA benefits from side-information about the source activations, for example via user guidance~\cite{Ono:2012bh} or pilot signals~\cite{Nesta:2017kn}.
As an example, the independent low-rank matrix analysis (ILRMA) framework successfully puts this principle to work and unifies IVA and NMF~\cite{Kitamura:2016vj}.
Whereas ILRMA applied a low-rank non-negative model on the separated source spectra, we propose instead to use the low-rank of the space-time sound power matrix as a proxy to the source activations.
The activations of the latent variables of the NMF model are assumed to be the variance of the separated source signals, effectively coupling together the IVA and NMF objectives.
This intuition is formalized as a joint probabilistic model of the sources and blinky signals, and we derive efficient updates to minimize its negative log-likelihood.

The performance of the algorithm is evaluated in numerical experiments and compared to that of AuxIVA~\cite{Ono:2011tn}.
The experiment results show that including the joint separation leads to improved performance in all tested cases.
Not only are the median SDR and SIR improved by up to 4 and 8 decibels (dB), respectively, but their variability is also significantly reduced, indicating stable performance.
In addition, we confirm that the use of extra microphones leads to steady improvement in performance, even for a weak source.

The rest of the paper is organized as follows.
In \sref{model}, we formulate the joint probabilistic model for the microphone and power sensor data.
An efficient algorithm for estimating the parameters of this model is described in \sref{algorithm}.
Results of numerical experiments validating the performance of the proposed method are given in \sref{numexp}.
\sref{conclusion} concludes this paper.

\begin{figure}
  \includegraphics[width=\linewidth]{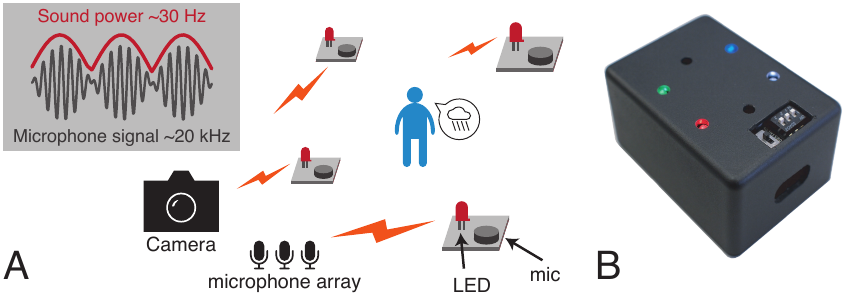}
  \caption{A) Example of a scenario with microphones and blinkies to cover a target source. B) Picture of an actual blinky sensor.}
  \flabel{blinkies}
\end{figure}

\section{Joint Model}
\seclabel{model}

We suppose there are $K$ target sources captured by $M$ microphones and $B$ sound power sensors.
In the short-time Fourier transform (STFT) domain, the microphone signals can be written as a weighted sum of the source signals 
\begin{equation}
  x_m[f,n] = \sum_{k=1}^K a_{mk}[f]\; y_k[f,n] + z_m[f,n]
\end{equation}
where $f=1,\ldots,F$ and $n=1,\ldots,N$ are the frequency and time indices, respectively.
The complex weight $a_{mk}[f]$ is the room transfer function from source $k$ to microphone $m$, and $z_m[f,n]$ collects the noise and model mismatch.
The $b$-th blinky signal at time $n$ is the sum of the sound power over frequencies at its location
\begin{equation}
  u_{bn} = \sum_{f=1}^F \left| \sum_{k=1}^K a_{bk}[f] \; y_k[f,n] + z_b[f,n] \right|^2.
\end{equation}
In addition, throughout the manuscript we use bold upper and lower case for matrices and vectors, respectively.
The Euclidean norm of a complex vector $\vx$ is denoted $\|\vx\| = (\vx^\H \vx)^{\half}$.

Our goal is to find the $M\times M$ demixing matrix $\mW_f$ such that the source signals are recovered linearly from the microphone measurements
\begin{equation}
  \vy_{fn} = \mW_f \vx_{fn}
\end{equation}
where
\begin{align}
  \vy_{fn} & = \left[y_1[f,n],\ldots,y_M[f,n]\right]^\top, \\
  \vx_{fn} & = \left[x_1[f,n],\ldots,x_M[f,n]\right]^\top, \\
  \mW_f & = [\vw_1\; \cdots\; \vw_M]^\H.
\end{align}
We will also overload notation in a natural way to represent the signal vector of source $k$ over frequencies
\begin{equation}
  \vy_{kn} = \left[y_k[1,n],\ldots,y_k[F,n]\right]^\top.
\end{equation}

We now establish the joint probabilistic model underpinning the algorithm we propose in \sref{algorithm}.
It is based on the three following assumptions.
\begin{enumerate}
  \item The separated signals $\vy_{kn}$ are statistically independent.
  \item The separated signals spectra are circularly-symmetric complex Normal random vectors with distribution
    \begin{equation}
      p_y(\vy_{kn}) = \frac{1}{\pi^F r_{kn}^F} \exp\left(- \frac{\| \vy_{kn} \|^2}{r_{kn}}\right), \quad k=1,\ldots,M,
    \end{equation}
    and time-varying variance $r_{kn}$. Taken over all time frames, this distribution is in fact super-Gaussian
    and has been successfully used for source separation~\cite{Ono:2012bh}.
  \item The power measurements $u_{bn}$ are the squared norms of circularly-symmetric complex Normal random vectors with covariance matrix $(\sum_{k=1}^K g_{bk} r_{kn}) \mI_F$, where $g_{bk}$ is a parameter of the power mix.
    Thus, the $B\times N$ non-negative matrix of the variances has rank $K$.
    The probability distribution function of the norm can be derived from the $\chi^2$ distribution with $2F$ degrees of freedom
    \begin{equation}
      p_u(u_{bn}) = \frac{1}{2^F \Gamma(F)} \frac{u_{bn}^{F-1}}{\sigma^{2F}}
      \exp\left( - \frac{u_{bn}}{2 \sigma^2} \right),
    \end{equation}
    with $\sigma^2 = \sum_{k=1}^K g_{bk} r_{kn}$.
    While this might seem like a deviation from usual Gaussian models, the same estimator of the variance is in fact obtained.
\end{enumerate}
Note that in the above we have maintained a distinction between the number of target sources $K$ and the number of microphones $M$.
For ICA and IVA, the determined case, i.e., $M=K$, needs to be assumed, and, because $\mW_f$ is an $M\times M$ demixing matrix, we will obtain $M$ demixed signals.
However, only $K$ out of $M$ sources are tied to the blinky signals via a low-rank non-negative variance model.
Intuitively, we are asking that the variances of these $K$ sources be well aligned with the activations from the non-negative decomposition.

Putting the pieces together, the likelihood of the observation is
\begin{equation}
  \calL = \prod_{f} |\det \mW_f|^{2N} \prod_{kn}p_y(\vy_{kn}) \prod_{bn} p_u(u_{bn}),
\end{equation}
with free parameters $\{\mW_f\}$, $\{g_{bk}\}$, $\{r_{kn}\}$.
The following section will describe how to estimate them  by minimizing the negative logarithm of this function.

\section{Algorithm}
\seclabel{algorithm}

In this section, we derive an algorithm to minimize the negative log-likelihood of the observed data.
The cost function derived can be written as the sum of those of IVA and NMF,
\begin{multline}
  J =
  -2 N \sum_f \log |\det \mW_f| + \sum_{n=1}^N\sum_{k=1}^M \left( \frac{\|\vy_{kn}\|^2}{r_{kn}}  + F \log r_{kn} \right) \\
  + \sum_{n=1}^N \sum_{b=1}^B \left( F \log \sum_{k=1}^K g_{bk} r_{kn} + \frac{u_{bn}}{2 \sum_{k=1}^K g_{bk} r_{kn}} \right) + C,
\elabel{cost_function}
\end{multline}
where $C$ includes all constant terms.
It is convenient to group the parameters to estimate in matrices.
We represent the gains by the matrix $\mG\in\R_+^{B\times K}$ with $(\mG)_{bn} = g_{bn}$ and the sources variance matrix by $\mR\in\R_+^{M\times N}$ with $(\mR)_{kn} = r_{kn}$.
Finally, we define $\mU \in \R_+^{B\times N}$ and $\mP \in \R_+^{K\times N}$ with $(\mU)_{bn} = u_{bn}$ and $(\mP)_{kn} = \| \vy_{kn} \|^2$, respectively.

The update rules for the demixing matrix $\mW_f$ are obtained from the iterative projection technique proposed for IVA~\cite{Ono:2011tn,Ono:2012bh}
\begin{align}
  \mV_{fk} & = \frac{1}{N} \sum_{n} \frac{1}{2 r_{kn}} \vx_{fn} \vx^\H_{fn}, \\
  \vw_{fk} & \gets (\mW_f \mV_{fk})^{-1} \ve_k, \\
  \vw_{fk} & \gets \vw_{fk} (\vw_{fk}^\H \mV_{fk} \vw_{fk})^{-\frac{1}{2}},
\end{align}
where $\vw_{fk}$ is the $k$-th row of $\mW_f$ and $\ve_k$ is the $k$-th canonical basis vector. These updates are done for $k=1$ to $M$.

\begin{figure}
  \centering
  \includegraphics[width=0.87\linewidth]{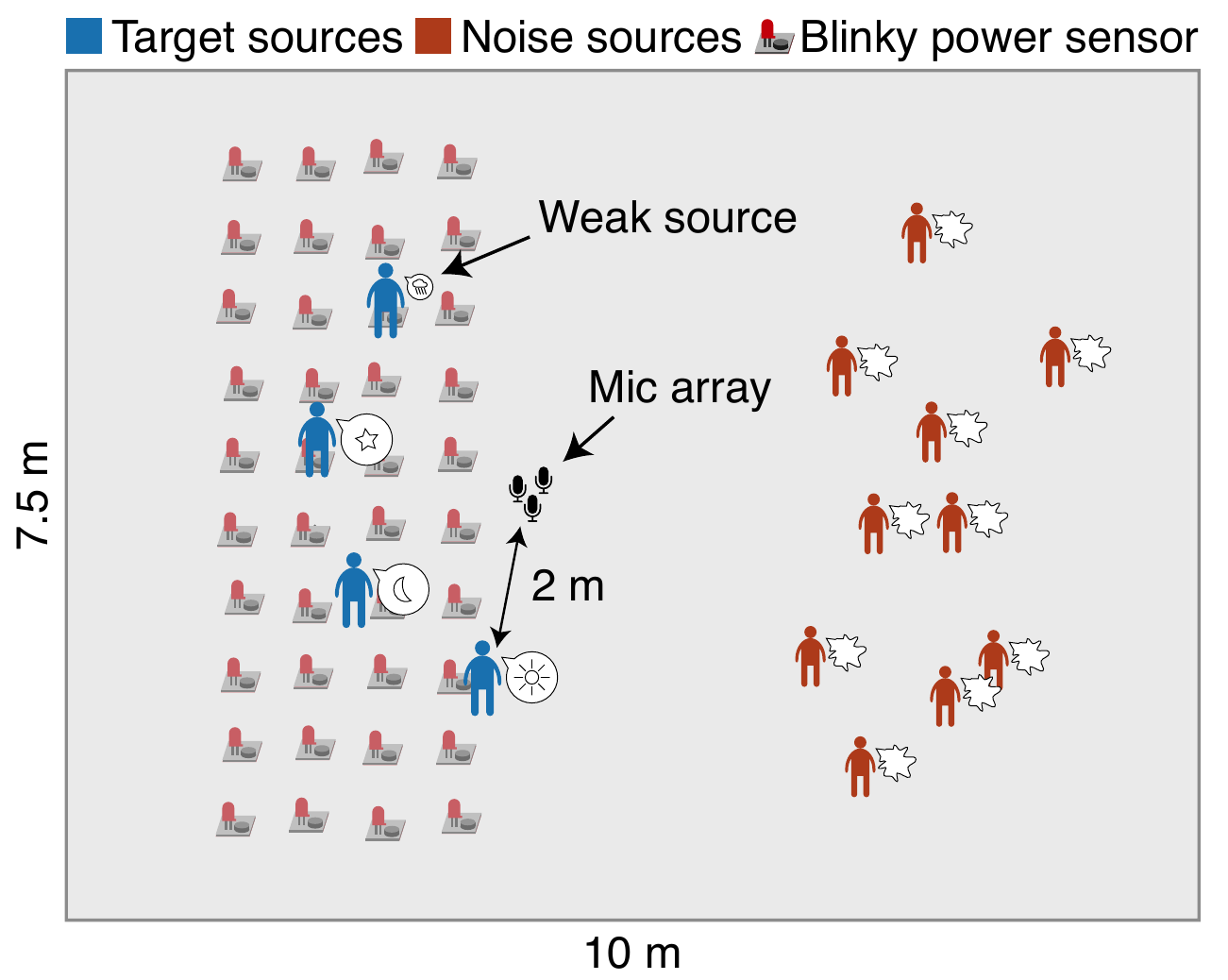}
  \caption{Illustration of the room geometry and locations of sources and sensors in the numerical experiments.}
  \flabel{experiment_setup}
\end{figure}

The update rules for $\mG$ and $\mR$ are very similar to those of IS-NMF and are given by the following proposition.
\begin{proposition}
  The following updates of $\mG$ and $\mR$ decrease monotonically the value of the cost function $J$ from \eref{cost_function}
  \begin{align}
    \nonumber
    \mG & \gets \mG \left( \frac{\left(\frac{1}{2F}\mU \odot (\mG\mR_K)^{.-2}\right) \mR_K^\top}{(\mG\mR_K)^{.-1} \mR_K^\top} \right)^{.\frac{1}{2}}, \\
    \nonumber
    \mR_K & \gets \mR_K \left( \frac{\frac{1}{F}\mP_K \odot \mR_K^{.-2} + \mG^\top \left(\frac{1}{2F}\mU \odot (\mG\mR_K)^{.-2}\right)}{\mR_K^{.-1} + \mG^\top (\mG\mR_K)^{.-1}} \right)^{.\frac{1}{2}}.
  \end{align}
  The dotted exponents, divisions, and $\odot$ are element-wise power, division, and multiplication operations, respectively.
  The matrices $\mR_K$ and $\mP_K$ contain the $K$ top rows of $\mR$ and $\mP$, respectively.
\end{proposition}
\begin{proof}
  Minimizing $J$ with respect to $\mG$ and $\mR_K$ is equivalent to minimizing the Itakura-Saito divergence $\calD_{\mathsf{IS}}(\wt{\mU}\;|\;\wt{\mG} \mR_K)$ with
  \begin{equation}
    \wt{\mU} = \frac{1}{F} \left[\; {\scriptstyle \frac{1}{2}} \mU^\top \  \mP_K^\top \;\right]^\top, \quad \wt{\mG} = \left[\; \mG^\top \  \mI_K \;\right]^\top.
  \end{equation}
  Then the above update rules are obtained by standard majorization-minimization of the $\beta$-divergence~\cite{MasahiroNakano:ey,Anonymous:S6ILS0MD}.
\end{proof}

When $K < M$, there are $M-K$ sources that are not coupled to the NMF part of the cost function.
The variance estimates of these sources is obtained by equating the gradient of \eref{cost_function} to zero, resulting in the following update
\begin{equation}
  r_{kn} = \frac{1}{F}\| \vy_{kn} \|^2, \quad k=K+1,\ldots,M,\ \forall\; n.
\end{equation}

There is an inherent scale indeterminacy between $\mW_f$, $\mR$, and $\mG$.
It is fixed by performing a normalization step after each iteration
\begin{equation}
  \begin{array}{l@{\;}l@{\quad}l@{\;}l}
    \mR & \gets \mLambda_M^{-1} \mR, & \mG & \gets \mG \mLambda_K, \\
    \mW_f & \gets \mLambda_M^{-\frac{1}{2}} \mW_f, & \mP & \gets \mLambda_M^{-1} \mP,
  \end{array}
  \elabel{rescaling}
\end{equation}
where $\mLambda_K = \frac{1}{N} \diag(\mR_K \mathds{1})$ is a diagonal matrix containing the average row values of $\mR$ up to row $K$ ($\mathds{1}$ is the all one vector).
These rescaling do not change the value of the cost function.
The full algorithm is summarized in \algref{blinkiva}.

\begin{algorithm}[t]
\SetKwInOut{Input}{Input}\SetKwInOut{Output}{Output}
\Input{Microphones $\{ \vx_{fn} \}$ and blinky signals $\{u_{bn}\}$}
\Output{Separated signals}
\DontPrintSemicolon
\For{loop $\leftarrow 1$ \KwTo $\text{max. iterations}$}{
  \nl \# Run a few iterations of NMF at once\;
  \For{loop $\leftarrow 1$ \KwTo $\text{nmf sub-iterations}$}{
    $\mR_K \gets \mR_K \left( \frac{\mP_K \odot \mR_K^{.-2} + \mG^\top \left({\scriptscriptstyle \frac{1}{2}}\mU \odot (\mG\mR_K)^{.-2}\right)}{F (\mR_K^{.-1} + \mG^\top (\mG\mR_K)^{.-1})} \right)^{.\frac{1}{2}}$\;
    $\mG \gets \mG \left( \frac{\left(\frac{1}{2F}\mU \odot (\mG\mR_K)^{.-2}\right) \mR_K^\top}{(\mG\mR_K)^{.-1} \mR_K^\top} \right)^{.\frac{1}{2}}$\;
  }
  \For{$k \leftarrow K+1$ \KwTo $M$}{
    \For{$n \leftarrow 1$ \KwTo $N$}{
      $r_{kn} \leftarrow \frac{1}{F} \| \vy_{kn} \|^2$\;
    }
  }
  \nl \# Update the demixing matrices\;
  \For{$k \leftarrow 1$ \KwTo $M$}{
    \For{$f \leftarrow 1$ \KwTo $F$}{
      $\mV_{fk} = \frac{1}{N} \sum_{n} \frac{1}{2 \max\{\epsilon, r_{kn}\}} \vx_{fn} \vx^\H_{fn}$\;
      $\vw_{fk} \gets (\mW_f \mV_{fk})^{-1} \ve_k$\;
      $\vw_{fk} \gets \vw_{fk} (\vw_{fk}^\H \mV_{fk} \vw_{fk})^{-\frac{1}{2}}$\;
    }
  }
  \nl \# Demix the signal\;
  \For{$f \leftarrow 1$ \KwTo $F$}{
    \For{$n \leftarrow 1$ \KwTo $N$}{
      $\vy_{fn} = \mW_f \vx_{fn}$\;
    }
  }
  \nl \# Rescale all the variables according to \eref{rescaling} \;
}
\caption{Joint separation of sources and sound power}
\label{alg:blinkiva}
\end{algorithm}

\section{Numerical Experiments}
\seclabel{numexp}

\begin{figure*}[h]
  \centering
  \subfloat{\flabel{perfeval:sdr} \includegraphics[width=90mm]{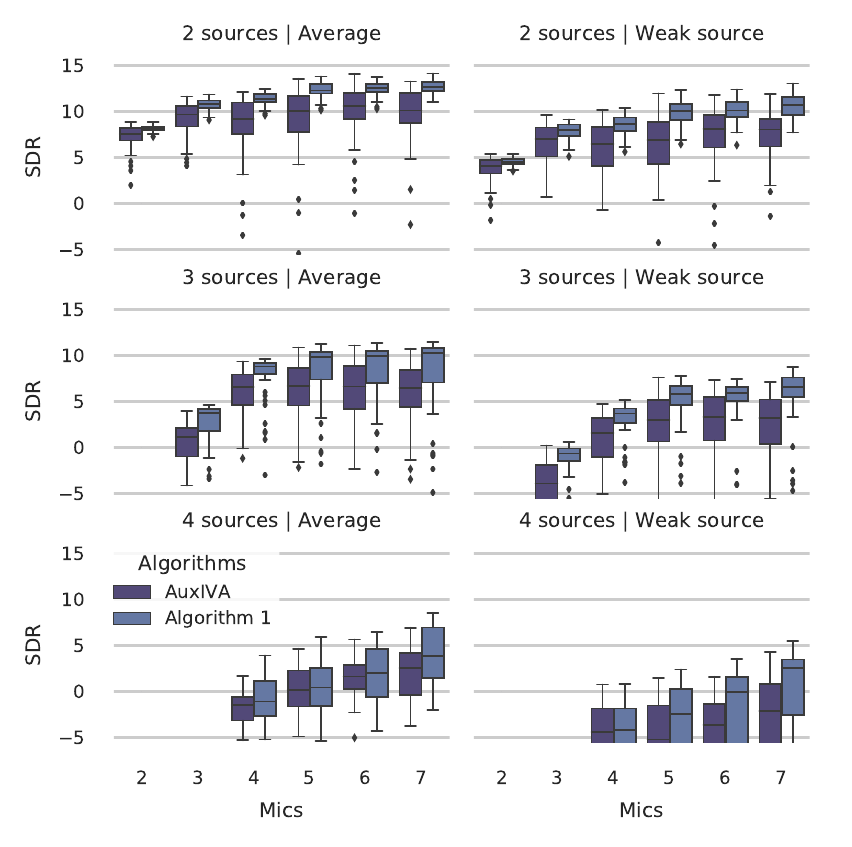}}
  \subfloat{\flabel{perfeval:sir} \includegraphics[width=90mm]{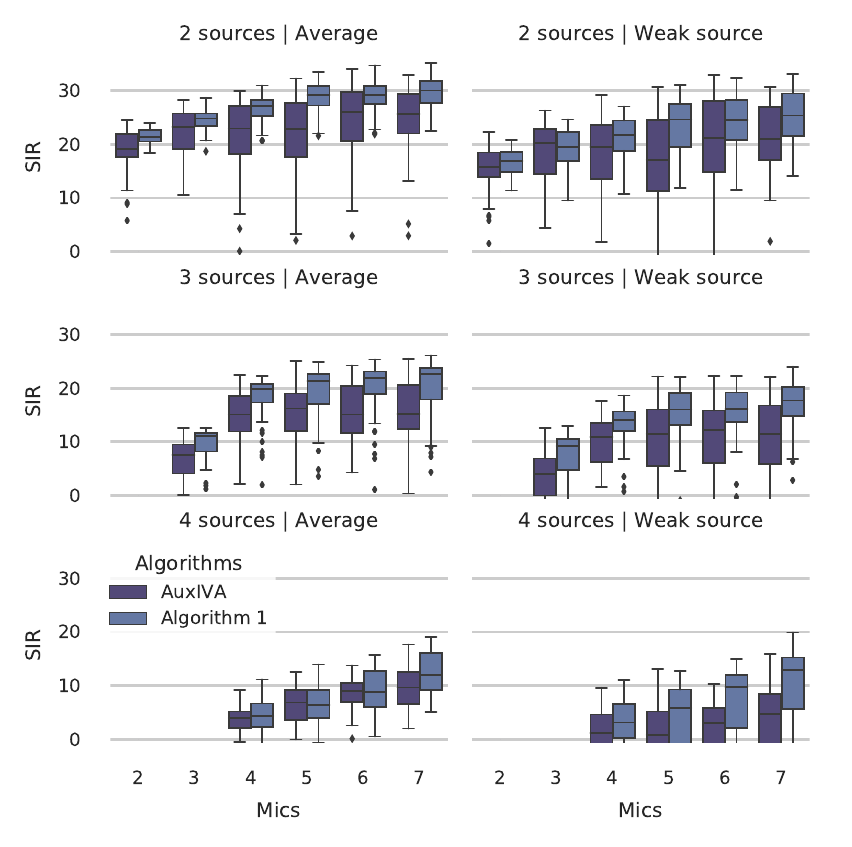}}
  \caption{
    Box-plots of signal-to-distortion ratio (SDR, left) and signal-to-interference ratio (SIR, right) of the separated signals.
    From top to bottom, the number of sources increases from 2 to 4.
    The number of microphones increases from 2 to 7 on the horizontal axis.
    Odd and even columns show results averaged over all sources and for the weak source only, respectively.
  }
  \flabel{perfeval}
\end{figure*}


In this section, we evaluate and compare the performance of \algref{blinkiva} to that of AuxIVA~\cite{Ono:2011tn} via numerical experiments.

\subsection{Setup}

We simulate a \SI{10}{\meter}$\times$\SI{7.5}{\meter}$\times$\SI{3}{\meter} room with reverberation time of \SI{300}{\milli\second} using the image source method~\cite{Allen:1979cn} implemented in \texttt{pyroomacoustics} Python package~\cite{scheibler2017pyroomacoustics}.
We place a circular microphone array of radius \SI{2}{\centi\meter} at $[4.1, 3.76, 1.2]$.
The number of microphones is varied from 2 to 7.
Forty blinkies are placed on an approximate $4\times 10$ equispaced grid filling the \SI{3}{\meter}$\times$\SI{5.5}{\meter} rectangular area with lower left corner at $[1,1]$.
Their height is \SI{0.7}{\meter}.
Between 2 and 4 target sources are placed equispaced on an arc of \SI{120}{\degree} of radius \SI{2}{\meter} centered at the microphone array.
They are placed at a height of \SI{1.2}{\meter} and such that they fall within the area covered by the grid of blinkies.
Diffuse noise is created by placing 10 additional sources on the opposite side of the target sources with respect to the microphone array.
The setup is illustrated in \ffref{experiment_setup}.

After simulating propagation, the variances of target sources are fixed to $\sigma_k^2 = 0.25$ for $k=1$ and $\sigma_k^2 = 1$ for $k \geq 2$.
The signal-to-noise and signal-to-interference-and-noise ratios are defined as
\begin{equation}
  \mathsf{SNR} = \frac{\frac{1}{K} \sum_{k=1}^K \sigma_k^2}{\sigma_n^2},\quad \mathsf{SINR} = \frac{\sum_{k=1}^K \sigma^2_k}{Q \sigma_i^2 + \sigma_n^2},
\end{equation}
where $\sigma_i^2$ and $\sigma_n^2$ are the variances of the $Q$ interfering sources and uncorrelated white noise, respectively.
We set them so that $\mathsf{SNR}=60$ dB and $\mathsf{SINR}=10$ dB.
Speech samples of approximately \SI{20}{\second} are created by concatenating utterances from the CMU Sphinx database \cite{Kominek:2004vf}.
All utterances in a sample are taken from the same speaker.
The experiment is repeated $50$ times for different attributions of speakers and speech samples to source locations.

The simulation is conducted at a sampling frequency of \SI{16}{\kilo\hertz}.
The STFT frame size is 4096 samples with half-overlap and uses a Hann window for analysis and matching synthesis window.
The $B$ blinky signals are simulated by placing extra microphones at their locations.
The blinky microphone signals are fed to the STFT and their power is summed over frequencies before processing\footnote{In practice, the blinky signals acquired via LEDs and a camera need to be calibrated and resampled at the STFT frame rate.}.
Finally, \algref{blinkiva} is compared to AuxIVA~\cite{Ono:2011tn} as implemented in \texttt{pyroomacoustics}~\cite{scheibler2017pyroomacoustics}.
Both algorithms are run for 100 iterations, and the number of NMF sub-iterations of \algref{blinkiva} is 20.
The scale of the separated signals is restored by projection back on the first microphone~\cite{Murata:2001gb}.

\subsection{Results}

We evaluate the separated signals in terms of signal-to-distortion ratio (SDR) and signal-to-interference ratio (SIR) as defined in~\cite{Vincent:2006fz}.
These metrics are computed using the \texttt{mir\_eval} toolbox~\cite{Raffel:2014uu}.
While \algref{blinkiva} provides automatic selection of the separated signals when $K < M$, this is not the case for AuxIVA.
As a work-around, we select the $K$ signals with the largest power for comparison.

The distribution of SDR and SIR of the separated signals is illustrated with box-plots in~\ffref{perfeval}.
Both the distribution averaged over all sources and for the weak source only are showed.
Overall, the joint formulation improves over AuxIVA in terms of both SDR and SIR improvements in all cases.
For two and three sources, while the performance of AuxIVA is very signal dependent, with dips as low as 0 dB in terms of SIR, the proposed method gives consistent performance around or above 20 dB.
Even for the weak source, the proposed method in many cases has a 25-th percentile higher than the 75-th percentile of AuxIVA.
With four sources, both methods have similar average performance when up to 6 microphones are used, with \algref{blinkiva} outperforming AuxIVA for 7 microphones.
In addition, we observe that only the proposed method successfully exploits extra microphones to extract the weak source.
With 7 microphones the SIR are 4.7 and 12.9 dB for AuxIVA and \algref{blinkiva}, respectively.

\section{Conclusion}
\seclabel{conclusion}

In this work, we showed that using sound power sensors, e.g. blinkies, together with a conventional microphone array significantly boosts the performance of blind source separation.
Because the blinkies can be distributed over a larger area, they can provide reliable source activation information.
We formulated a joint probabilistic model of the microphone and power sensor measurements and used it to derive an efficient algorithm for the blind extraction of sources from the microphone signals.
We showed through numerical experiments that including the power data effectively regularizes the source separation.
Performance of the proposed method increases steadily with the number of microphones, unlike conventional IVA that suffers in some cases from frequency permutation ambiguity.
In addition, the proposed method is able to recover a source with just a quarter the power of three competing sources, whereas conventional IVA fails to do so.

A key question that has yet to be answered is that of the influence of the placement of blinkies with respect to sound sources.
Indeed, we would like to determine the minimum density and under what conditions the joint separation performs best.
Finally, the proposed algorithm should be tested in real conditions.



\bibliographystyle{IEEEtran}
\bibliography{refs}

\end{document}